\documentclass[reqno,A4paper]{amsart}
\usepackage{graphics,epsfig}
\usepackage[latin1]{inputenc}
\usepackage[english,activeacute]{babel}
\usepackage[all]{xy}
\usepackage{amsmath, amsthm, amsfonts, amssymb}

\theoremstyle{plain}
\newtheorem{theorem}{Theorem}[section]
\newtheorem{prop}[theorem]{Proposition}
\newtheorem{cor}[theorem]{Corollary}

\newtheorem*{theorem*}{Theorem}
\theoremstyle{definition}
\newtheorem{defs}[theorem]{Definition}
\newtheorem{example}[theorem]{Example}
\newtheorem{remark}[theorem]{Remark}

\begin{document}

\title[States in generalized probabilistic models]
{States in generalized probabilistic models: an approach based in algebraic geometry}

\author[C\'esar Massri\and Federico Holik\and \'Angelo Plastino]
{{\sc C\'esar Massri}$^{2,4}$ \ {\sc ,} \ {\sc Federico Holik}$^{1,3}$ \ {\sc
and} \ {\sc \'Angelo Plastino}$^{1}$}

\address{\llap{1 - \,}Universidad Nacional de La Plata, Instituto
de F\'{\i}sica (IFLP-CCT-CONICET), C.C. 727, 1900 La Plata, Argentina}
\address{\llap{2 - \,}Departamento de Matem\'{a}tica -
Universidad CAECE - Buenos Aires, Argentina}
\address{\llap{3 - \,}Center Leo Apostel for Interdisciplinary Studies and, Department
of Mathematics, Brussels Free University Krijgskundestraat 33, 1160 Brussels, Belgium}
\address{\llap{4 - \,}Instituto IMAS, Buenos Aires, CONICET}

\begin{abstract}
We present a characterization of states in generalized
probabilistic models by appealing to a  non-commutative
version of geometric probability theory based on algebraic geometry
techniques. Our theoretical framework allows for  incorporation of invariant states in
a natural way.
\end{abstract}


\thanks{This work has been supported by CONICET (Argentina).}

\maketitle

\section{Introduction}

States in physical theories must determine well defined
probabilities for any testable empirical event. In classical
theories, such as classical statistical mechanics, probabilistic
states can be represented by measurable functions on a phase space.
In this way, states of classical theories obey the axioms of
Kolmogorov for probabilities. In Kolmogorov's setting, probabilities
are considered as measures defined over a $\sigma$-algebra of
measurable subsets of a given outcome set. Kolmogorov's achievement
solved a problem posed in the year 1900 by the Mathematician David
Hilbert \cite[p. 454]{HilbertProblems-1902} (Hilbert's sixth problem)
with regard to the axiomatization of probabilities. The ensuing
solution  was given in terms of measure theory, being
suitable for a rigorous treatment of the problem.

Quantum mechanics showed that this correspondence between
states and Kolmogorovian measures is not universal.
Quantum states are represented by density operators
\cite{Holik-Ciancaglini,Holik-Zuberman-2013}, which are
very different mathematical objects. Indeed, a quantum
state, defines a Kolmogorovian probability distribution
for each empirical setup \cite{Svozil-Contexts-2009}. But
if we mix events from different and incompatible setups,
then Kolmogorov's rules for probability can be shown to
fail \cite{GudderQP}. Thus, the best we can do is to
represent a quantum state as a \emph{family} of
Kolmogorovian measures pasted in a harmonic way
\cite{Holik-NaturalInformationMeasures}. Note that
probabilities given by measurements on a single empirical
setup are not enough: we must perform measurements of
different incompatible observables in order to determine
a quantum state completely (see for example
\cite{Tomography-RevModPhys}).

This can be expressed by using other words: while the set
of states of a classical probabilistic theory forms a
simplex, the sets of quantum states exhibit a more
involved geometrical structure \cite{Bengtsson}. This
distinction acquires a very specific mathematical
expression: while classical probabilistic states can be
described as measures over Boolean lattices, quantum
states are described by measures over non-Boolean ones.
This fact is related  to the non-commutative character of
observables. The algebras involved in models of standard
quantum mechanics can differ from those used in quantum
field theory \cite{HaggLQP}, or in quantum statistical
mechanics \cite{Bratteli,Bratteli2}. While in
non-relativist quantum mechanics we only have Type I
factors, other factors can appear in a relativistic
setting and in quantum statistical mechanics
\cite{HalvorsonARQFT,RedeiSummersQuantumProbability,Bratteli}.
The algebras associated to classical systems are
commutative \cite{Bratteli}. The rigorous formulation of
quantum theory can be considered as an extension of
measure theory to the non-commutative setting
\cite{Hamhalter-QuantumMeasureTheory}. States in quantum
theories can be also characterized as non-Kolmogorovian
probabilistic measures
\cite{RedeiSummersQuantumProbability}. These formal
developments can be considered as a continuation of the
program outlined in Hilbert's sixth problem.

In many situations of interest, the task of the physicist is to find
out which is the state of the system under study, given that certain
constraints are specified. One of the most important constraints are
given by symmetries, and symmetries are usually represented
mathematically by the action of groups. Thus, the problem of
determining measures which are invariant under the action of certain
groups is an important one for physics. The problem of finding
invariant states in operator algebras was studied in the
mathematical physics literature, and  is a hard one (see for
example \cite{Bratteli}, Section $4.3$).

In this work, we study this problem in the setting of general
orthomodular lattices. Our results can be connected with those of
the above mentioned algebraic approach in a natural way, because the
algebraic structure of some of the operator algebras used in
mathematical physics is in close relation with the properties of the
lattices formed by their self-adjoint idempotent elements (i.e.,
projection operators)
\cite{KalmbachOrthomodularLatticesBook,RingsOfOperatorsI,
RingsOfOperatorsII,RingsOfOperatorsIII,RingsOfOperatorsIV}.
Applications of probabilities defined over non-Boolean structures
are not restricted to physical theories (see for example
\cite{Khrennikov-2010} and \cite{Yukalov-Sornette-2017}). The case
of measures over bounded lattices is studied separately
\cite{Massri-Holik-2017} and should not be confused with the study
of states in probabilistic theories. Our contribution is twofold

\begin{itemize}
\item From the mathematical point of view, we provide an extension
of the mathematical framework involved in the derivation of
Groemer's integral theorem in a non-Boolean setting. This could be
considered as a non-commutative generalization of geometric
probability theory based on algebraic geometry techniques.
\item We apply our extension to quantum theories represented by event
structures which can be non-Boolean in general. This allows us, via
our non-Boolean (or non-commutative) version of Groemer's theorem,
to characterize all states which are
invariant under the action of a general group $G$.
\end{itemize}

In order to find a suitable framework for incorporating conditions
imposed by the action of groups, in \cite{HolikIJGMMP} we proposed a
reformulation of the MaxEnt problem in terms of a non-Boolean
version of \emph{geometric probability theory}
\cite{Rota1998,RotaBook}. This rephrasing of the problem can be
helpful for an axiomatic approach to the study of symmetries in the
quantum setting, and thus, it can also be considered as a
continuation of the program outlined by Hilbert's sixth problem. The
present work is a natural follow up  of the ideas presented in
\cite{HolikIJGMMP}.

Our formalism could be helpful in contexts such as the ones
appearing in the algebraic formulation of quantum statistical
mechanics or quantum field theory. It is also suited for the
application of the MaxEnt method to a wide range of models, because
it is compatible with the specification of linear constraints (or
more general ones) representing mean values \cite{HolikIJGMMP}.

\begin{description}
\item[Outline of the paper]
In Section \ref{s:Preliminaries} we review  elementary notions
and  start by the description of states in physical
theories, and other mathematical preliminaries. We place emphasis on
the fact that states of classical theories can be considered as
measures over Boolean algebras, while states of quantum theories
must appeal to measures over orthomodular lattices.

In Section \ref{s:RingOfFunctions} we discuss our approach for the
particular case of complete Boolean atomistic lattices. This example is very
important, because classical probabilistic models can be described
by measures over $\sigma$-algebras which are Boolean lattices. The
constructions and methods of this Section help to illuminate the
extensions to the non-Boolean case of the rest of the paper. The
content of this Section can be considered as a physical version of geometric probability theory.

In Section \ref{s:ortho} we extend the construction of Section
\ref{s:RingOfFunctions} to the more general case of additive
orthogonal measures on orthocomplemented lattices. In this Section
we apply our general theoretical framework to the kind of measures
which appear in physical theories of interest (such as classical
statistical mechanics, non-relativistic quantum mechanics, algebraic
quantum field theory and algebraic quantum statistical mechanics).
In this way, we show how our theoretical framework is useful to
characterize invariant states. We also discuss the case of Gleason's
theorem
\cite{Hamhalter-QuantumMeasureTheory,Dvurecenskij-Gleason,Dvurecenskij-2009}.

The mathematical intuition behind our construction is that of
extending geometric probability theory to the non-Boolean case, by
appealing to algebraic geometry techniques. The physical intuition
is that of considering physical states as invariant measures (see
\cite{HolikIJGMMP} for a previous version of this approach). The
cases of interest for physics are given by orthomodular lattices and
measures valued in the ring of real numbers. But we build our
mathematical framework for the general case first, because in this
way it is easier to understand the essential mathematical structures
underlying our construction. A key feature of our theoretical framework is that
we appeal to a mathematical technique based in algebraic geometry.
Starting with an arbitrary orthocomplemented lattice $\mathcal{L}$,
we build the abelian group $\mathbb{Z}^{\oplus \mathcal{L}}$
generated by $\mathcal{L}$, formed by all possible formal sums of
its elements. By taking the ratio
$M(\mathcal{L})=\mathbb{Z}^{\oplus\mathcal{L}}/S$ with respect to a
suitably chosen subgroup $S$, we obtain a commutative group which
will be useful to characterize invariant measures\footnote{This
construction is standard in commutative algebra. See for example the
construction of the tensor product in \cite[Prop.2.12]{MR0242802}.}.
We show that a linear function from $M(\mathcal{L})$ is the same as
a measure on $\mathcal{L}$. One of our main results is given by
Theorem \ref{orth-groe}, which allows us to characterize, by means
of Corollary \ref{c:Corollary1}, all invariant measures acting by
automorphisms on the lattice. This means that, given an arbitrary
lattice endowed with a set of measures (which could represent states
of a given physical model), and an arbitrary group (which could be a
group of automorphisms representing a physical symmetry), we obtain
a canonical and constructive characterization of its invariant
measures. We end Section \ref{s:ortho} with some examples regarding
positive measures and Gleason's theorem.

Finally, in Section \ref{s:Conclusions}, we draw some conclusions.

\end{description}

\section{Preliminaries}\label{s:Preliminaries}

In this Section we introduce some elementary notions on lattice theory and states of physical theories as
measures over lattices. The reader familiarized with lattice theory,
may skip this.

\subsection{Lattice Theory}

A poset $(\mathcal{L},\leq)$ is called a \emph{lattice}
if every two elements $x,y$ have a supremum $x\vee y$ and
an infimum $x\wedge y$
\cite{KalmbachOrthomodularLatticesBook}. We denote by
$\mathbf{0}$ and $\mathbf{1}$ to the bottom and top
elements, respectively, of a bounded lattice
$(\mathcal{L},\leq,\vee,\wedge)$. A bounded lattice
$(\mathcal{L},\leq,\vee,\wedge)$ is called
\emph{complemented} if for every $x\in \mathcal{L}$ there
exists $y\in \mathcal{L}$ such that
$$x\vee y=\mathbf{1},\quad x\wedge y=\mathbf{0}.$$

\noindent A morphism between bounded lattices
$f:(\mathcal{L},\leq,\vee,\wedge,\mathbf{0},\mathbf{1})\rightarrow
(\mathcal{L}',\leq,\vee,\wedge,\mathbf{0},\mathbf{1})$ is
a poset function such that
$$f(x\wedge_\mathcal{L} y)=f(x)\wedge_{\mathcal{L}'} f(y),\quad f(x\vee_\mathcal{L} y)=f(x)\vee_{\mathcal{L}'} f(y),\quad f(\mathbf{0})=\mathbf{0} ,\quad f(\mathbf{1})=\mathbf{1}.$$

\noindent An element $z$ of a poset is an \emph{atom} if the
interval $\{\mathbf{0}\leq a\leq z\}$ consists on two elements,
\[
\{\mathbf{0}\leq a\leq z\}=\{\mathbf{0},z\}.
\]

\noindent A lattice $\mathcal{L}$ is called \emph{atomic}
if for every $x\in \mathcal{L}\setminus \mathbf{0}$,
there exists an atom $z$ such that $z\leq x$. It is said
to be \emph{atomistic} if every element can be written as
a join of atoms. A lattice $\mathcal{L}$ is called
\emph{distributive} if for every $x,y,z\in \mathcal{L}$,
\[
x\vee(y\wedge z)=(x\vee y)\wedge (x\vee z),\quad x\wedge(y\vee
z)=(x\wedge y)\vee (x\wedge z).
\]

\noindent A complemented distributive lattice is called a
\emph{Boolean} lattice. In a distributive lattice,
complements are unique. A lattice in which every
countable subset has a meet and a join is called a
\emph{$\sigma$-lattice}. The next Proposition is
going to be useful in Section \ref{s:RingOfFunctions}.

\begin{prop}\label{dif2}
If $z$ is an atom of a distributive lattice $\mathcal{L}$, then
$$z\leq a\vee b\Longleftrightarrow z\leq a\text{ or }z\leq b.$$
\end{prop}

\begin{proof}
First, note that $z\wedge b$ and $z\wedge a$ are less than or equal
to $z$, hence they must be equal to $0$ or to $z$. If $z\wedge
b=0=z\wedge a$, from the distributive law, $z\wedge(a\vee b)=0$.
But, from the equation $z\leq a\vee b$, we obtain $z=z\wedge(a\vee
b)$, a contradiction. The other implication is obvious.
\end{proof}

An \emph{orthocomplementation} on a bounded lattice is a function that maps each element $a$
to a complement $a^\bot$ in such a way that the following axioms are satisfied:
\begin{itemize}
\item Complement: $a \vee a^\bot = 1$ and $a\wedge a^\bot = 0$.
\item Involution: $(a^\bot)^\bot = a$.
\item Order-reversing: if $a \leq b$ then $b^\bot \leq a^\bot$.
\end{itemize}

\noindent An \emph{orthocomplemented} lattice is a bounded lattice
which is equipped with an orthocomplementation. Orthocomplemented
lattices, like Boolean algebras, satisfy de Morgan's laws:
\[
(a\vee b)^\bot = a^\bot \wedge b^\bot,\quad (a\wedge b)^\bot = a^\bot \vee b^\bot.
\]
A morphism between orthocomplemented lattices
is a lattice map that preserves the orthocomplementation.
We say that a group $G$ acts on a orthocomplemented lattice
if there is a group map from $G$ to the group of automorphisms
of the lattice.

\

\noindent An \emph{orthomodular} lattice $\mathcal{L}$ is an
orthocomplemented lattice such that

\[
a\leq b\Longrightarrow a\vee (a^\bot\wedge b) = b,\quad\forall a,b\in \mathcal{L}.
\]

\noindent We say that $x$ is \emph{orthogonal} to $y$, denoted
$x\bot y$, if $y\leq x^\bot$.

\subsection{States as measures}\label{s:StatesAsMeasures}

A general approach to physical theories, very useful for
discussing foundational purposes, is based on von-Neumann
algebras
\cite{HalvorsonARQFT,RedeiSummersQuantumProbability}. Let
$\mathcal{B}(H)$ be the algebra of bounded linear
operators acting on a separable Hilbert space $H$. A
\emph{von Neumann algebra} is a $\ast$-subalgebra of
$\mathcal{B}(H)$ that contains the identity operator and
that is closed in the weak operator topology. The
\emph{weak operator topology} is defined as the weakest
topology such that the functional sending an operator $A$
to the complex number $\langle A\phi,\psi\rangle$, is
continuous for all vectors $\phi,\psi\in\mathcal{H}$. Due
to the celebrated von Neumann's double commutant theorem
\cite{mikloredeilibro}, von Neumann algebras can be
regarded as a $\ast$-subalgebra $\mathcal{W}$, containing
the identity operator and satisfying
$\mathcal{W}''=\mathcal{W}$, where given
$S\subseteq\mathcal{B}(H)$, $S'$ is defined as

\begin{equation}
S'=\{A\in\mathcal{B}(H)\,\,|\,\,AB-BA=0\,\forall\,B\in S\}
\end{equation}

\noindent $\mathcal{B}(H)$ is a von Neumann algebra, but
it is not the only one. For example, it is possible to
consider the algebra of observables of a classical theory
as a commutative von Neumann algebra. In algebraic
relativistic quantum field theory, the algebras
associated to space time regions will be generally the
class of von Neumann algebras known as Type III factors
\cite{Yngvason2005-TypeIIIFactors}. A similar situation
occurs in the algebraic approach to quantum statistical
mechanics, where Type II and Type III factors may appear
\cite{RedeiSummersQuantumProbability}. Thus, the approach
based on von Neumann algebras is adequate for a very
general family of physical models. A \emph{state} $\nu
:\mathcal{W}\longrightarrow\mathbb{C}$ is defined as a
continuous positive linear functional such that
$\nu(\mathbf{1})=1$ (see Reference
\cite{RedeiSummersQuantumProbability}\footnote{Throughout
this paper we assume that $\mathcal{W}$ is unital, i.e.,
that it possesses an identity element $\mathbf{1}$}).
Positivity means that $\nu(A^{\ast} A)\geq 0$ for all
$A\in\mathcal{W}$ or, equivalently, that $\nu(A)\geq 0$
for all $A\geq 0$, see \cite[\S 2.1]{MR866671}.
\emph{Normal states} can be defined as those states
satisfying the condition
$\nu(\sup_{\alpha}(A_{\alpha}))=\sup_{\alpha}\nu(A_{\alpha})$
for any uniformly bounded increasing net ${A_{\alpha}}$
of positive elements of $\mathcal{W}$ (equivalently,
states satisfying $\nu(\sum_{i\in I}E_{i}) = \sum_{i\in
I}\nu(E_{i})$ for any countable and pairwise orthogonal
family of projection operators\footnote{Here, a
\emph{projection operator} is understood as a self
adjoint and idempotent element of the algebra.
\emph{Elementary tests} to be performed on the system
involved, also known as \emph{events}, are represented by
projection operators (c.f. \cite{mikloredeilibro}).}
$\{E_{i}\}_{i\in I}$). Normal states are real-valued for
self-adjoint elements of $\mathcal{W}$ (recall that $A$
is self adjoint iff $ A^{\ast}=A$) and they form a convex
set. It can be shown that the set of projection operators
associated to an arbitrary von Neumann algebra
$\mathcal{W}$ is an orthomodular lattice
$\mathcal{L}(\mathcal{W})$ \cite{mikloredeilibro}. Thus,
normal states of physical theories define probabilities
on orthomodular lattices satisfying the following
properties:

\begin{subequations}\label{e:GeneralizedProbability}
\begin{equation} \label{e:GenProb1}
    \nu:\ \mathcal{L} \rightarrow [0,1],
\end{equation}
such that
\begin{equation}
    \nu(\mathbf{1})=1, \\
\end{equation}
and, for a denumerable and pairwise orthogonal family of events
$\{E_{i}\}_{i\in I}$,
\begin{equation} \label{e:GenProb3}
    \nu(\sum_{i\in I}E_{i}) = \sum_{i\in I}\nu(E_{i}) \,.
\end{equation}
\end{subequations}

\noindent where $\mathcal{L}$ is the lattice associated to a
particular algebra. In standard quantum mechanics $\mathcal{L}$ is
the lattice of projection operators on a Hilbert space and Gleason's
theorem ensures that it defines a density operator. In classical
mechanics $\mathcal{L}$ is the Boolean lattice of measurable subsets
of phase space. In algebraic relativistic quantum field theory, each
local region and a global state will define a lattice (associated to
a Type III factor) satisfying the above axioms. Notice that Eqs.
\ref{e:GeneralizedProbability} make sense for mathematical objects
more general than orthomodular lattices \cite{GudderQP}. For
example, one may use orthocomplemented lattices, or more generally,
orthocomplemented posets (which are not necessarily lattices)
\cite{GudderQP}. Another important remark is that if, instead of
condition \ref{e:GenProb3}, we only impose additivity, we will
obtain the more general notion of \emph{additive probability}:

\begin{subequations}\label{e:GeneralizedProbability2}
\begin{equation} \label{e:GenProb12}
    \nu:\ \mathcal{L} \rightarrow [0,1],
\end{equation}
such that
\begin{equation}
    \nu(\mathbf{1})=1, \\
\end{equation}
and, for orthogonal elements $E_{1}$ and $E_{2}$,
\begin{equation} \label{e:GenProb32}
    \nu(E_{1}\vee E_{2}) = \nu(E_{1})+\nu(E_{2}) \,.
\end{equation}
\end{subequations}

\noindent In this work, we will refer to probabilities defined by
Eqs. \ref{e:GeneralizedProbability} and
\ref{e:GeneralizedProbability2} as \emph{$\sigma$-additive} and
\emph{additive}, respectively. In fact, we will change the codomain
of $\nu$ in conditions (\ref{e:GeneralizedProbability}a) and
(\ref{e:GeneralizedProbability2}a) in order to work with
\emph{measures}. See
\cite{Hamhalter-QuantumMeasureTheory,Dvurecenskij-1985} for
generalizations of Gleason's theorem.

It is  important to notice that the Boolean vs.
non-Boolean distinction is crucial here: while states of
classical theories are measures over Boolean lattices
(and thus, they obey Kolmogorov's axioms), states of
non-classical theories cannot be equated to Kolmogorovian
measures. Furthermore, all quantal theories ``suffer" of
contextuality: the Kochen-Specker theorem can be
generalized to arbitrary von Neumann algebras
\cite{Doring2005}. This crucial distinction allows one to
see that quantum theories define a new form of measure
theory, much more general than the one used in
Kolmogorov's axioms
\cite{RedeiSummersQuantumProbability,Hamhalter-QuantumMeasureTheory}.
In particular, as we show in our work, the whole
theoretical framework must be adapted in order to deal
with invariant measures representing physical symmetries.

A measure defined in an orthomodular lattice satisfying
Eqns. \ref{e:GeneralizedProbability} can be considered as
a family of Kolmogorovian measures: if the lattice is
non-Boolean, this family has only one member (the Boolean
lattice itself). But if the lattice is not Boolean, it is
possible to show that an orthomodular lattice can be
considered as a pasting of maximal Boolean subalgebras
\cite{navara1991pasting}, and a state will define a
Kolmogorovian measure when restricted to one of them
\cite{Holik-NaturalInformationMeasures,Svozil-Contexts-2009}.

The following Definition formalizes the idea of a state which
possesses a symmetry represented by the action of a group:

\begin{defs}
Given an algebra $\mathcal{W}$ and a group $G$, denote the action of
$G$ by $A\in\mathcal{W}$ as $\tau_{g}(A)\in\mathcal{W}$ for all $g
\in G$. A state $\nu$ is said to be $G$-invariant if
$\nu(\tau_{g}(A))=\nu(A)$ for all $A\in\mathcal{W}$ and for all
$g\in G$.
\end{defs}

\noindent Notice that if $G$ is represented by automorphisms in the
lattice of projection operators, similar definitions can be made for
states defined as measures in Eqns. \ref{e:GeneralizedProbability}
and \ref{e:GeneralizedProbability2}.

\

Suppose now that the group $G$ acts on an orthomodular
lattice $\mathcal{L}$ by automorphisms. Denote by $\alpha_{g}$ to
the automorphism corresponding to $g\in G$. Thus, a measure $\nu$
satisfying Equations \ref{e:GeneralizedProbability}, is said to be
\emph{invariant} under the action of $G$ if $\nu(\alpha_{g}x)=\nu(x)$, for
all $x\in\mathcal{L}$ and all $g\in G$, see \cite{Varadaran}.

\section{The ring of functions on a complete Boolean atomistic lattice}\label{s:RingOfFunctions}

In this Section, we restrict ourselves to complete Boolean atomistic lattices in order to
illustrate on some properties that will be extended later to the non-Boolean
case. Let us  start by recalling the
description of a classical statistical system. States of a classical
system will be described by points in a measurable phase space
$\Gamma$. Observables are described by functions
$f:\Gamma\longrightarrow \mathbb{R}$ such that for every measurable
set $E\subseteq\mathbb{R}$, $f^{-1}(E)$ is measurable in $\Gamma$,
usually $\Gamma=\mathbb{R}^{2N}$ endowed with the Lebesgue measure.
The appeal to measure theory allows one to build well defined
probabilistic states: testable events will be defined as measurable
subsets of $\Gamma$. As it is well known, the measurable subsets of
$\Gamma$ form a Boolean lattice (thus, states can be defined as in
Eqns. \ref{e:GeneralizedProbability}, with $\mathcal{L}$ equal to the
lattice of measurable subsets of $\Gamma$).

Let us show how to assign a ring to a
complete Boolean atomistic lattice. Let $\mathcal{B}$ be
a complete Boolean atomistic lattice and let
$\mathcal{A}(\mathcal{B})$ be the set of its atoms. For
every $x\in\mathcal{B}$, we can consider the
\emph{indicator function} of $x$,
$I_x:\mathcal{A}(\mathcal{B})\rightarrow\mathbb{R}$ as

\begin{equation}
I_x(z)=
\begin{cases}
1&\text{if } z\leq x,\\
0&\text{if not.}
\end{cases}
\end{equation}

\noindent Using the properties of
$\mathcal{B}$ as an atomistic lattice and Prop.
\ref{dif2} of Section \ref{s:Preliminaries}, it is easy
to see that the indicator functions satisfy for all
$x,y,x_1,\ldots,x_n\in \mathcal{B}$,

\begin{eqnarray}\label{e:IndicatorProperties}
&I_xI_y=I_{x\wedge y}&\nonumber\\ &I_{x\vee y}+I_{x\wedge
y}=I_x+I_y&\nonumber\\ &I_{x_1\vee\ldots\vee
x_n}=1-(1-I_{x_1})\ldots(1-I_{x_n})
\end{eqnarray}

\bigskip

\noindent Let us call $Fun(\mathcal{A}(\mathcal{B}))$ the
algebra of functions from $\mathcal{A}(\mathcal{B})$ to
$\mathbb{R}$. The algebra operations are
given by pointwise addition and multiplication. The real
numbers act by scalar multiplication. Inside
$Fun(\mathcal{A}(\mathcal{B}))$ we have the subalgebra $Sim(\mathcal{A}(\mathcal{B}))$ of
simple functions generated by the indicators (i.e., all
possible finite combinations of the form $\sum_{i=1}^n
\alpha_i I_{x_i}$). Notice that we can
assign to every $x\in\mathcal{B}$ the indicator function
of its atoms,
\[
\mathcal{B}\to Sim(\mathcal{A}(\mathcal{B})),\quad
x\mapsto I_x:=I_{\{z\in\mathcal{A}(\mathcal{B})\,\colon\,z\leq x\}}.
\]
Given that $\mathcal{B}$ is atomistic, this assignment is injective and by
Prop. \ref{dif2}, $I_{x\vee y}=I_{\{z\leq x\}\cup \{z\leq y\}}$.
From the completeness of $\mathcal{B}$ is follows that any
indicator function $I_A$ is equal to $I_x$, $x=\bigvee_{z\in A} z$.

\

\noindent Notice that a linear functional $\tilde{\nu}$ from
$Sim(\mathcal{A}(\mathcal{B}))$ to $\mathbb{R}$ must satisfy

\begin{equation}\label{e:NuMono}
\tilde{\nu}(I_{x\vee y})+\tilde{\nu}(I_{x\wedge
y})=\tilde{\nu}(I_x)+\tilde{\nu}(I_y)
\end{equation}

\bigskip

\noindent An important idea is at stake here. Given a Boolean
lattice $\mathcal{B}$, we say that a function
$\nu:\mathcal{B}\rightarrow \mathbb{R}$ is an \emph{additive
measure}\footnote{Compare with the general definition given by Eqns.
\ref{e:GeneralizedProbability2}. The main difference relies on the
fact that probabilistic measures such as the ones defined by
\ref{e:GeneralizedProbability2} are normalized (its range is
restricted to the interval $[0,1]$), while this last condition is
not assumed for additive measures.} (or just a \emph{measure}) if
given $x,y\in\mathcal{B}$,

\[
\nu(x\vee y)+\nu(x\wedge y)=\nu(x)+\nu(y).
\]

\noindent The measure $\nu$ is called \emph{bounded} if it has a bounded
image, \emph{positive} if it takes
positive values. It is called a \emph{probability function} if its
image is the interval $[0,1]$ (and thus, it is a particular case of
Eqs. \ref{e:GeneralizedProbability2}, for $\mathcal{L}$ a Boolean
lattice). As it is well known, the above definition of measure can be
derived from the following axiom,

\[
\nu(x\vee y)=\nu(x)+\nu(y),\quad \forall x\bot y.
\]

\noindent The advantage of this axiom is that it can be used to
extend the definition of additive measure to any orthocomplemented
lattice $\mathcal{L}$. If we instead impose the condition

\[
\nu(\bigvee_{i=0}^\infty x_{i})=\sum_{i=0}^\infty \nu(x_{i}),
\quad x_i\bot x_j,\,\forall i\neq j,
\]

\noindent we say that $\nu$ is a \emph{$\sigma$-additive measure} on
$\mathcal{L}$\footnote{Again, compare with the general Definition on
arbitrary orthomodular lattices given by Eqs.
\ref{e:GeneralizedProbability}.}.

\bigskip

\noindent

Let us go back to our previous constructions.
We noticed that a linear functional $\tilde{\nu}$ from
$Sim(\mathcal{A}(\mathcal{B}))$ to $\mathbb{R}$ can be viewed
(by restriction to $\mathcal{B}$)
as a measure on the complete Boolean atomistic lattice $\mathcal{B}$. This
connection is provided by Eq. \ref{e:NuMono}. Indeed, given a
function $\tilde{\nu}$ defined on
$Sim(\mathcal{A}(\mathcal{B}))$, we obtain a measure on
$\mathcal{B}$ by defining $\nu(x)=\tilde{\nu}(I_{x})$, for all
$x\in\mathcal{B}$. Eq. \ref{e:NuMono} guarantees that $\nu$ satisfies
the axioms of an additive measure.
Analogously, a measure $\nu$ on $\mathcal{B}$ determines a
linear functional $\tilde{\nu}$ from
$Sim(\mathcal{A}(\mathcal{B}))$
to $\mathbb{R}$ by defining $\tilde{\nu}(I_{x})=\nu(x)$.
By Eqs. \ref{e:IndicatorProperties} we obtain the desired functional.
Specifically, we proved a bijection,
\[
\phi:\{\nu:\mathcal{B}\to\mathbb{R}\,\colon\,
\nu(x\vee y)+\nu(x\wedge y)=\nu(x)+\nu(y)\}
\to
\{\tilde{\nu}:Sim(\mathcal{A}(\mathcal{B}))\to\mathbb{R}\,\colon\,
\tilde{\nu}\text{ linear}\}.
\]
If $\phi(\nu)=\phi(\nu')$,
then $\phi(\nu)(I_x)=\phi(\nu')(I_x)$ for all $x\in\mathcal{B}$.
Hence, $\nu=\nu'$. Given $\tilde{\nu}$, set $\nu=\tilde{\nu}|_{\mathcal{B}}$.
Then, $\phi(\nu)=\tilde{\nu}$.
This mathematical fact
has a simple physical explanation (when probabilities are
considered): the probabilities are determined by their values in
elementary cases. The simpler functional is perhaps the
evaluation map, $ev_z(I_x)=I_x(z)$, which corresponds to measuring the
property $z$.

\

If a group $G$ acts\footnote{
The action is defined as a group map $G\to\text{Aut}(\mathcal{B})$,
where $\text{Aut}(\mathcal{B})$ is the set (group) of
bijective morphisms of complete Boolean atomistic lattices, $\mathcal{B}\to\mathcal{B}$,
that is, preserving $0,1,\wedge,\vee$ and the atoms.
}
on $\mathcal{B}$,
it follows that an invariant
measure can be made
equivalent to an invariant linear functional
$Sim(\mathcal{A}(\mathcal{B}))\rightarrow\mathbb{R}$.
Indeed, given an
invariant measure $\nu:\mathcal{B}\to\mathbb{R}$, we define
$\tilde{\nu}:Sim(\mathcal{A}(\mathcal{B}))\to
\mathbb{R}$ as $\tilde{\nu}(I_{x})=\nu(x)$.
It is easy to prove that the action induces a linear action
of $G$ in $Sim(\mathcal{A}(\mathcal{B}))$ by
$g\cdot I_x=I_{g\cdot x}$. Hence,
$\tilde{\nu}$ is invariant if and only if $\nu$ is invariant.

\

Let us characterize $\sigma$-additive measures on
$\mathcal{B}$. Denote by
$Fun^b(\mathcal{A}(\mathcal{B}))$ to the space of bounded
functions with respect to the supremum norm
$\|f\|=\text{sup}_{\mathcal{A}(\mathcal{B})}|f(x)|$. It
is a non-separable metric space and $\|I_x\|=1$ for all
$x\in \mathcal{B}$. Let us denote by
$Sim(\mathcal{A}(\mathcal{B}))^{\ast}$ to the space of
continuous linear functionals on
$Sim(\mathcal{A}(\mathcal{B}))$.

Let $\nu$ be a $\sigma$-additive measure. Then, $\tilde{\nu}$ defines
a linear functional $Sim(\mathcal{A}(\mathcal{B}))\to\mathbb{R}$.
Let $I_{x_n}$ be a convergent sequence of indicator functions such that
$\lim_n I_{x_n}=I_{x}$. Then,
there exists an orthogonal sequence $y_n\in\mathcal{B}$ such that $x_n=y_1\vee\ldots\vee y_n$.
Then,
\[
x_n=y_1\vee\ldots\vee y_n\Longrightarrow
\tilde{\nu}(I_{x_n})=
\tilde{\nu}(I_{y_1\vee\ldots\vee y_n})=
\sum_{i=1}^n\tilde{\nu}(I_{y_i})=
\sum_{i=1}^n\nu(y_i)
\Longrightarrow
\]
\[
\lim_n\tilde{\nu}(I_{x_n})=
\lim_n\sum_{i=1}^{n}\nu(y_i)=
\nu(\lim_n\bigvee_{i=1}^{n}y_i)=
\tilde{\nu}(\lim_n I_{x_n})=
\tilde{\nu}(I_x).
\]
Analogously, if $\tilde{\nu}$ is linear continuous,
take $\{y_i\}$ an orthogonal sequence, then
\[
\nu(\bigvee_{i=1}^\infty y_i)=
\tilde{\nu}(\lim_nI_{\bigvee_{i=1}^n y_i})=
\tilde{\nu}(\lim_n\sum_{i=1}^n I_{y_i})=
\lim_n\sum_{i=1}^n \tilde{\nu}(I_{y_i})=
\sum_{i=1}^\infty\nu(y_i).
\]
Summing up, we have a bijection between
$\sigma$-additive measures on $\mathcal{B}$
and $Sim(\mathcal{A}(\mathcal{B}))^{\ast}$.
Using the same arguments as before, there exists a bijection between
$\sigma$-additive invariant
measures on $\mathcal{B}$ and invariant continuous linear functionals
on $Sim(\mathcal{A}(\mathcal{B}))$.

\subsection{Generating sets}

\noindent It is the case that many
lattices of interest in physics exhibit the property of being
atomistic: each element can be expressed as the join of atoms.
Examples of lattices with this property are points of phase space for a
classical system and the set of one dimensional projections in
standard quantum mechanics. This is a very useful property, but in
the general case, lattices can be non-atomic (and thus,
non-atomistic), for example, the continuous geometries studied by
von Neumann \cite{KalmbachOrthomodularLatticesBook}. In order to
generalize the constructions presented in the previous Section, the
following definitions will be useful:

\begin{defs}
Let $\mathcal{B}$ be a Boolean lattice and let $B\subseteq
\mathcal{B}$ be a subset closed under finite meets. We say that $B$
is a
\begin{itemize}
\item \emph{generating set} if every $x\in \mathcal{B}$ satisfies
$x=b_1\vee\ldots\vee b_k$ for some $b_1,\ldots,b_k\in B$.
\item $\sigma$-\emph{generating set} if every $x\in \mathcal{B}$ is equal
to $x=\bigvee b_i$ for some denumerable subset $\{b_i\}\subseteq B$.
\end{itemize}
\end{defs}

\noindent Before we continue, it is important to make the following
remark. We can still characterize measures
on a lattice (atomic or not) as measures on a manifold by using a suitably chosen
generating set. This is the power of Groemer's integral theorem:

\begin{theorem}[Groemer's integral theorem]
Let $\mathcal{B}$ be a distributive $\sigma$-lattice and let $B$ be
a $\sigma$-generating set for $\mathcal{B}$. Let
$\nu:B\rightarrow\mathbb{R}$ be a function satisfying $\nu(0)=0$ and
$\nu(b_1\vee b_2)+\nu(b_1\wedge b_2)= \nu(b_1)+\nu(b_2)$ for all
$b_{1},b_{2},b_{1}\vee b_{2}\in B$. Then, the following statements
are equivalent

\begin{enumerate}
\item $\nu$ extends uniquely to a measure on $\mathcal{B}$.
\item $\nu$ satisfies the inclusion-exclusion identities
\[
\nu(b_1\vee\ldots\vee b_k)=\sum_{i}\nu(b_i)
-\sum_{i<j}^k\nu(b_i\wedge b_j)+\ldots
,\quad\forall b_1,\ldots,b_k,b_1\vee\ldots\vee b_k\in B,\, k\geq 2.
\]
\end{enumerate}
\end{theorem}
\begin{proof}
See \cite[Th.2.2.1]{MR1670882}
or \cite[Th.1]{MR0513905}.
\end{proof}

\noindent Although Groemer's Theorem yields a strong result, we want to
extend it in two directions. The first one is to a situations in which  a group is
acting and the second is to orthomodular lattices. To
illustrate these ideas, let us consider the following examples. The
first one was extracted from \cite{MR1670882}.

\begin{example}
A polyconvex set in $\mathbb{R}^n$ is a finite union of compact
convex sets of dimension $n$. Union and intersections of polyconvex
sets are polyconvex sets, hence the family of polyconvex sets is a
non-atomic Boolean lattice. In this lattice, the set of
parallelotopes $[0,x_1]\times\ldots\times[0,x_n]$
combined with the group of Euclidean motions (translations and
rotations) is a $\sigma$-generating set.
Let us denote the Euclidean group as $G$
and $B$ to the set of parallelotopes of the form $[0,x_1]\times\ldots\times[0,x_n]$,
$B\cong\mathbb{R}^n$.
Notice that the
the \emph{normalizer} subgroup of $B$,
$N_G(B)=\{g\in G\,\colon\,gB=B\}$,
is the group of permutations $S_n$,
\[
g\cdot [0,x_1]\times\ldots\times[0,x_n]= [0,x_{g\cdot
1}]\times\ldots\times[0,x_{g\cdot n}]\in B,\quad g\in S_n.
\]
\noindent According to \cite[Ch.5, Ch.4]{MR1670882}, there is a
bijection between $G$-invariant measures and $S_n$-invariant
functions $B=\mathbb{R}^n\rightarrow\mathbb{R}$ satisfying the
inclusion-exclusion identities.
\end{example}

In the next Subsection we extend the above constructions over
orthomodular lattices (under the action of more general groups). The
following example illustrates what happens in a non-distributive
lattice:

\begin{example}
If $\mathcal{L}$ is the (non-Boolean)
lattice of subspaces of $\mathbb{R}^{n}$, the manifold of atoms is
given by $\mathbb{P}^{n-1}$ and the general linear group $G$ acts
on $\mathcal{L}$. Fix a line $\langle v\rangle$.
We prove in the next Section that any
$G$-invariant measure $\nu$ on $\mathcal{L}$ is determined by
assigning a real number
$\alpha\in\mathbb{R}$ to the line $\langle v\rangle$.
In fact, the value of $\nu$ over a $k$-dimensional subspace
$L=\langle v_1,\ldots,v_k\rangle$ is $k\alpha$,
\[\nu(L)=\nu(\langle v_1\rangle)+\ldots+\nu(\langle v_k\rangle)=k\alpha.\]
Then, invariant measures over this lattice
are essentially the dimension functions.
\end{example}

\section{Invariant measures on arbitrary orthocomplemented lattices.}\label{s:ortho}

Let us now extend the results of the previous Section to more
general lattices. As seen in Section
\ref{s:StatesAsMeasures}, the need of extending the approach to a
more general setting is related to the fact that quantum theories
are based on non-commutative algebras for which the event structures
are non-distributive (and thus, non-Boolean). The category of
orthomodular lattices is large enough to include many physical
examples of interest \cite{kalm83}.

\subsection{Measures over orthocomplemented lattices}

To understand the set of measures on a particular lattice in the
light of our constructions, it is more convenient to use category
theory. We introduce the notion of a measure with values in an abelian group $A$.
Next, we show that this defines a functor (a type of mapping between categories) that is representable,
i.e., there is a group \emph{parameterizing}, the set of measures.

\begin{defs}
Let $\mathcal{L}$ be an orthocomplemented lattice and let $A$ be an  abelian group.
\begin{itemize}
\item An \emph{additive measure} on $\mathcal{L}$ with values in $A$ is a function
$\nu:\mathcal{L}\rightarrow A$ satisfying
\[
\nu(x\vee y)=\nu(x)+\nu(y),\quad\text{if }x\bot y.
\]
\item
If $A$ is complete, a \emph{$\sigma$-additive measure} on
$\mathcal{L}$ with values in $A$ is a function
$\nu:\mathcal{L}\rightarrow A$ satisfying
\[
\nu(\bigvee x_i)=\sum_{i=1}^\infty\nu(x_i),\quad\text{if }x_i\bot x_j,\forall i\neq j\in\mathbb{N}.
\]
where the term in the right is a convergent series in $A$.
\end{itemize}

\noindent It is important to remark that in Eqs.
\ref{e:GeneralizedProbability} and
\ref{e:GeneralizedProbability2} of Section
\ref{s:StatesAsMeasures}, i) the abelian group is taken
to be the set of  real numbers, ii) the lattices are
considered orthomodular, and iii) the measures are
normalized to the interval $[0,1]$ (this  is the
environment of interest for describing states in physical
theories).

Let us denote by $\mathcal{M}(\mathcal{L};A)$  the space
of measures with values in $A$. If a group $G$ acts by
lattice automorphisms in $\mathcal{L}$, an
\emph{invariant measure} is a measure $\nu$ on
$\mathcal{L}$ such that $\nu(g\cdot x)=\nu(x)$ for all
$x\in\mathcal{L}$ and $g\in G$. The space of invariant
measures is denoted $\mathcal{M}(\mathcal{L};A,G)$. By
definition
$\mathcal{M}(\mathcal{L};A)=\mathcal{M}(\mathcal{L};A,0)$,
where $0$ is the trivial group. We always assume that the
action of $G$ in $A$ is trivial.
\end{defs}

\

\begin{remark}
In what follows, we are going to use some results and constructions
from category theory and commutative algebra. First, let us denote
by $\mathbf{Ab}$  the category of abelian groups. We say that a
functor $F:\mathbf{Ab}\to\mathbf{Ab}$ is \emph{representable} if
there exists $M\in \mathbf{Ab}$ such that
$F(A)=\text{Hom}_{\mathbb{Z}}(M,A)$. Representable functors are an
important tool  nowadays in mathematics. This is due to the fact that
they allow us to represent structures taken from an abstract (or not
completely understood) category in terms of sets (or abelian groups) and functions.
Usually, these notions are easier to handle, allowing for i) a better
understanding of the original structures, and eventually, ii) to the
application of known techniques to more abstract
domains.

Secondly, we are going to use the $\otimes$-$\text{Hom}$ adjunction
that will be important in what follows.
Let $K_1,K_2$ be two commutative rings and let $M$ a $K_1$-module,
$N$ a $K_2$-module, and $T$ a $(K_1,K_2)$-bimodule. Then, there
exists an isomorphism of abelian groups

\[
\text{Hom}_{K_1}(N\otimes_{K_2} T,M)\cong
\text{Hom}_{K_2}(N,\text{Hom}_{K_1}(T,M)).
\]

Third, given a set $L$ we can construct the \emph{free abelian
group generated by} $L$, $\mathbb{Z}^{\oplus L}$. It is given by all
finite formal sums of elements of $L$, $\{\sum_{i=1}^k l_i\}$.
Another (equivalent) characterization is given by the set of all
functions with finite support $L\to \mathbb{Z}$. The function
associated to an element $l\in L$ is the indicator function
$I_{\{l\}}$. We show below that the group $\mathbb{Z}^{\oplus L}$
has the following \emph{universal property}: any function $\nu:L\to
A$ extends uniquely to a $\mathbb{Z}$-linear map $\overline{\nu}:\mathbb{Z}^{\oplus
L}\to A$:

\[
\xymatrix{
L\ar[r]^{\forall\nu}\ar[d]_{\pi}
\ar@{}[dr]|<<<<\circlearrowleft&A\\
\mathbb{Z}^{\oplus L}\ar@{-->}[ur]_{\exists!\overline{\nu}}&
}
\]

Finally, given an abelian group $M$ where a group $G$
acts, we can construct the module of invariants $M^G$ and the module
of coinvariants $M_G$,
\[
M^G=\{m\in M\,\colon\, g\cdot m=m,\,\forall g\in G\},\quad
M_G=M/R.
\]
where $R$ is the submodule generated by
$\{g\cdot m-m\,\colon\,g\in G,\,m\in M\}$. The module $M_G$ has the
following universal property: for every linear map $\nu:M\to A$,
where $G$ acts trivially on $A$, there exists a unique factorization
$\overline{\nu}:M_G\to A$,
\[
\xymatrix{
M\ar[r]^{\forall\nu}\ar[d]_{\pi}
\ar@{}[dr]|<<<<\circlearrowleft&A\\
M_G\ar@{-->}[ur]_{\exists!\overline{\nu}}&
}
\]

\end{remark}

Now we apply all these notions to study different measures of
interest on orthocomplemented lattices:

\begin{prop}\label{rep}
Let $\mathcal{L}$ be an orthocomplemented lattice.
Then, there exists an abelian group $M=M(\mathcal{L})$ representing the functor
$\mathcal{M}(\mathcal{L};-)$,
\[
\mathcal{M}(\mathcal{L};-)=\text{Hom}_{\mathbb{Z}}(M,-)
\]
\noindent This means that a measure in $\mathcal{L}$ valued in $A$ is equivalent to a $\mathbb{Z}$-linear map from $M(\mathcal{L})$ to $A$.
\end{prop}

\begin{proof}
Let $\mathbb{Z}^{\oplus \mathcal{L}}$ be the free abelian group
generated by $\mathcal{L}$. Let $S\subseteq
\mathbb{Z}^{\oplus \mathcal{L}}$ be the subgroup generated by the
elements
\[
\{x\vee y-x-y\,\colon\,x\bot y\}.
\]
where the elements of $\mathcal{L}$ are embedded in
$\mathbb{Z}^{\oplus \mathcal{L}}$ as usual (i.e., by using their
characteristic functions). Let us see that $M=\mathbb{Z}^{\oplus
\mathcal{L}}/S$ represents the functor
$\mathcal{M}(\mathcal{L};-)$.

First, assume that $x_i\bot x_j$ for
every $1\leq i<j\leq k$. Given that $x_i\leq x_1^\bot$ for $2\leq
i\leq k$, we have $x_2\vee\ldots\vee x_k\leq x_1^\bot$ and then,
$x_1\vee(x_2\vee\ldots\vee x_k)=x_1+(x_2\vee\ldots\vee x_k)$ in $M$.
By induction, we obtain that $x_1\vee\ldots\vee x_n=x_1+\ldots+x_n$
in $M$.

Second, any measure $\mathcal{L}\rightarrow A$ defines a
function $\mathbb{Z}^{\oplus \mathcal{L}}\rightarrow A$ mapping
$S$ to zero. Analogously, any linear function $\nu:M\rightarrow A$
is a linear function from $\mathbb{Z}^{\oplus \mathcal{L}}$ such
that $\nu(S)=0$. Then, the restriction of $\nu$ to $\mathcal{L}$
defines a measure on $\mathcal{L}$.
\end{proof}

\begin{cor}
Let $\mathcal{L}$ be an orthocomplemented lattice and assume that a group $G$ acts by automorphism. Let $A$ be an abelian group
where $G$ acts trivially.

The map $\pi:\mathcal{L}\rightarrow M(\mathcal{L})$ is a measure and
satisfies the following universal property: any measure with values
in $A$ factorizes as a $\mathbb{Z}$-linear map
$M(\mathcal{L})\rightarrow A$.

Also, the map $\pi_G:\mathcal{L}\rightarrow M_G$ is an invariant measure
and satisfies the following universal property:
any invariant measure with
values in $A$ factorizes as a $\mathbb{Z}$-linear
map $M_G\rightarrow A$.

We can represent the properties of $\pi$ and $\pi_G$
by the following diagrams:
\[
\xymatrix{
\mathcal{L}\ar[r]^{\forall\nu}\ar[d]_{\pi}
\ar@{}[dr]|<<<<\circlearrowleft&A\\
M(\mathcal{L})\ar@{-->}[ur]_{\exists!\overline{\nu}}&
}\qquad
\xymatrix{
\mathcal{L}\ar[r]^{\forall\nu'}\ar[d]_{\pi_G}
\ar@{}[dr]|<<<<\circlearrowleft&A\\
M(\mathcal{L})_G\ar@{-->}[ur]_{\exists!\overline{\nu'}}&
}
\]
\noindent where $\nu$ (resp. $\nu'$) is a measure (resp. invariant
measure) and $\overline{\nu}$, $\overline{\nu'}$ are linear maps.
The commutativity means that $\nu=\overline{\nu}\pi$,
$\nu'=\overline{\nu'}\pi_G$.
\end{cor}
\begin{proof}

\noindent Clearly, $\pi$ is a measure with values in $M$
and $\pi_G$ is an invariant measure with values in $M_G$.
The result for $\pi$ follows from Proposition \ref{rep}.

\noindent Now recall that the functor $(-)_G$ is naturally isomorphic to $(-)\otimes_{\mathbb{Z}[G]}\mathbb{Z}$ and the functor $(-)^G$ is naturally isomorphic to $\text{Hom}_{\mathbb{Z}[G]}(\mathbb{Z},-)$, \cite[Lemma 6.1.1]{MR1269324}.
Then, using the $\otimes$-$\text{Hom}$ adjunction and the equality $A=A^G$,
we see that any map $\overline{\nu'}\in\text{Hom}_{\mathbb{Z}}(M,A)$
is in fact a $G$-linear map.
Next, recall that any $G$-linear map $M\rightarrow A$ into a trivial $G$-module $A$
factors uniquely over a map from $M_G$.
Hence, $\overline{\nu'}$ factorizes uniquely
over $M_G$.

Summing up, given $\nu'\in\mathcal{M}(\mathcal{L};A,G)$
there exists a unique $\overline{\nu'}\in\text{Hom}_{\mathbb{Z}}(M_G,A)$
such that $\overline{\nu'}\pi=\nu'$.
\end{proof}

\

\begin{defs}
We call $\pi$ the \emph{universal measure} and $\pi_G$ the
\emph{universal invariant measure}. \noindent When no confusion
arises, we write just $x$ instead of $\pi(x)$ (or $\pi_G(x)$).
\end{defs}

\

The following Corollaries summarize our constructions up to now:

\begin{cor}\label{c:Corollary1}
We have the following characterizations,
\[
\mathcal{M}(\mathcal{L};A)=\text{Hom}_\mathbb{Z}(M,A),
\quad
\mathcal{M}(\mathcal{L};A,G)=\text{Hom}_\mathbb{Z}(M,A)^G.
\]
Moreover, if $K$ is a commutative ring and $A$ is a $K$-module
(ex. $K=\mathbb{R}$ and $A=\mathbb{R}^n$),
\[
\mathcal{M}(\mathcal{L};A)=\text{Hom}_K(M_K,A),\quad
\mathcal{M}(\mathcal{L};A,G)=\text{Hom}_K(M_K,A)^G,
\]
where $M_K:=M\otimes_{\mathbb{Z}} K$ is the extension
of $M$ from $\mathbb{Z}$ to $K$. Also, the map $\pi_K:\mathcal{L}\to M_K$
is the universal $K$-linear measure.
\end{cor}
\begin{proof}
Using the $\otimes$-$\text{Hom}$ adjunction, we have
\[\text{Hom}_{\mathbb{Z}}(M_G,A)=\text{Hom}_{\mathbb{Z}}(\mathbb{Z}\otimes_{\mathbb{Z}[G]}M,A)=\]
\[\text{Hom}_{\mathbb{Z}[G]}(\mathbb{Z},\text{Hom}_\mathbb{Z}(M,A))=\text{Hom}_\mathbb{Z}(M,A)^G.\]

\noindent Now, if $A$ is a $K$-module,
from the $\otimes$-$\text{Hom}$ adjunction we have,
\[
\text{Hom}_K(M_K,A)=\text{Hom}_K(M\otimes_\mathbb{Z}K,A)=
\text{Hom}_\mathbb{Z}(M,\text{Hom}_K(K,A))=\text{Hom}_\mathbb{Z}(M,A).
\]
\end{proof}

\begin{cor}
Let $K$ be a commutative ring. Then,
$\mathcal{M}(-;K)$ is a functor from the category of orthocomplemented lattices to the category of $K$-modules.

Also, if we fix the orthocomplemented lattice,
$\mathcal{M}(\mathcal{L};-)$ is a functor
from the category of $K$-modules to the category of $K$-modules.
\end{cor}
\begin{proof}
Let $f:\mathcal{L}\rightarrow \mathcal{L}'$ be a lattice map between two orthocomplemented lattices.
Pre-composing with $f$, we get a $K$-linear map, $f^*:\mathcal{M}(\mathcal{L}';K)\to \mathcal{M}(\mathcal{L};K)$.

Analogously, if $\phi:A\rightarrow A'$ is $K$-linear,
$\phi_*:\mathcal{M}(\mathcal{L};A)\rightarrow \mathcal{M}(\mathcal{L};A')$ is also $K$-linear.
\end{proof}

\

\noindent Let us mention some simple applications of our result. For
example, if we are interested in measures with values over
$\mathbb{C}$, our results show that they are parameterized by
complex linear functionals. Also, if we are interested in (for
example) $G$-invariant real measures, we proved that they are
parameterized by $G$-invariant real linear maps.

\subsection{Generalization of Groemer's integral theorem}

In this Subsection we extend Groemer's integral theorem to the more
general setting of orthocomplemented lattices. For this purpose, we
need to generalize the notion of \emph{generating set} to arbitrary
orthocomplemented lattices:

\begin{defs}
Let $\mathcal{L}$ be an orthocomplemented lattice and assume that a
group $G$ acts on the lattice $\mathcal{L}$ by automorphisms. Let
$B\subseteq \mathcal{L}$ be a subset closed under finite meets.
The group $G$ may not act on $B$ but we always has
an action of its normalizer subgroup,
\[N_G(B)=\{g\in G\,\colon\, gB=B\}.\]
Notice that the inclusion $B\subseteq\mathcal{L}$ induces a
set function
(not necessarily injective) $\mu:B/N_G(B)\to\mathcal{L}/G$.
We say that $B$ is
\begin{itemize}
\item
an \emph{orthogonal generating set} if every $x\in \mathcal{L}$ is
equal to $x=b_1\vee\ldots\vee b_k$ for some pairwise orthogonal
$b_1,\ldots,b_k\in B$.
\item
a $\sigma$-\emph{orthogonal generating set} if every $x\in
\mathcal{L}$ is equal to $x=\bigvee b_i$ for some pairwise
orthogonal numerable subset $\{b_i\}\subseteq B$.
\item an \emph{orthogonal generating set for the action} if
$\mu$ is injective and
the set $G\cdot B$ is an orthogonal generating set for $\mathcal{L}$.
\item a $\sigma$-\emph{orthogonal generating set for the action} if
$\mu$ is injective and
the set $G\cdot B$ is a $\sigma$-orthogonal generating set for $\mathcal{L}$.
\end{itemize}
\end{defs}

\

\noindent Now we are ready to construct our extension of Groemer's
integral theorem to orthocomplemented lattices:

\begin{theorem}[Non-Boolean Groemer's integral theorem]\label{orth-groe}
Let $\mathcal{L}$ be an orthocomplemented lattice where a group $G$ acts.
Let $B$ be an orthogonal generating set for the action of $G$.
Then, invariant measures on $\mathcal{L}$ are in bijection with
$N_G(B)$-invariant functions on $B$, $\nu$, such that
\[
\nu(b_1\vee b_2)=\nu(b_1)+\nu(b_2),\quad
\forall b_1,b_2,b_1\vee b_2\in B,\quad b_1\bot b_2.
\]
\end{theorem}
\begin{proof}
Let a)  $S\subseteq\mathbb{Z}^{\oplus\mathcal{L}}$ be the subgroup
generated by $\{x\vee y-x-y\,\colon\,x\bot y\}$ and
b) $\overline{B}\subseteq\mathcal{L}/G$ be the image
of $\mu:B/N_G(B)\to\mathcal{L}/G$.
Recall that $M=M(\mathcal{L})\cong \mathbb{Z}^{\oplus\mathcal{L}}/S$.
By hypothesis, the following map is surjective,
\[
\tilde{\mu}:\mathbb{Z}^{\overline{B}}\to M_G,\quad \overline{b}\mapsto [b].
\]
Then, $\mathbb{Z}^{\overline{B}}/\ker(\tilde{\mu})\cong M_G$, where
\[
\ker(\tilde{\mu})=\mathbb{Z}^{\overline{B}}\cap S_G=
\langle
\overline{b_1\vee b_2}-\overline{b_1}-\overline{b_2}\,\colon\,
\overline{b_1},\overline{b_2},\overline{b_1\vee b_2}
\in \overline{B},\, \overline{b_1}\bot\overline{b_2}\rangle.
\]
Finally, applying the functor $\text{Hom}_\mathbb{Z}(-,A)$, we obtain
\[
\text{Hom}_{\mathbb{Z}}(M_G,A)
\cong
\text{Hom}_{\mathbb{Z}}(\mathbb{Z}^{\oplus \overline{B}}/\ker(\tilde{\mu}),A).
\]
The first space is isomorphic to $\mathcal{M}(\mathcal{L};A,G)$.
A map $\nu$ in the second space
is a map in
$\text{Hom}_{\mathbb{Z}}(\mathbb{Z}^{\oplus \overline{B}},A)$,
such that $\nu(\ker(\tilde{\mu}))=0$. More specifically,
it is the same as a function
$\nu:\overline{B}\to A$ such that
\[
\nu(\overline{b_1\vee b_2})=\nu(\overline{b_1})+\nu(\overline{b_2}),\quad
\overline{b_1},\overline{b_2},\overline{b_1\vee b_2}
\in \overline{B},\quad \overline{b_1}\bot\overline{b_2}.
\]
The result follows from $B/N_{G}(B)\cong\overline{B}$.
\end{proof}

\

\begin{remark}
A similar result holds for $\sigma$-additive measures.
The steps are the following:
\begin{enumerate}
\item Define the set
$\mathcal{M}^{\sigma}(\mathcal{L};A)$
of $\sigma$-additive
measures with values in a Banach vector space $A$
over a complete field $K$ (say $K=\mathbb{R}$).
\item
Given that $M_{K}=M\otimes_{\mathbb{Z}}K$
is complete, see \cite[Ch.10]{MR0242802},
we can define the set of continuous linear maps (i.e. bounded operators) from
$M_{K}$ to $A$, that is,
$\nu\in \text{Cont}_{K}(M_{K},A)$ if and only if
\[
\nu(\bigvee_{i=1}^\infty x_i)=\sum_{i=1}^\infty\nu(x_i),\quad x_i\bot x_j,\,\forall i\neq j.
\]
From the following commutative square, it follows that
$\mathcal{M}^{\sigma}(\mathcal{L};A)\cong
\text{Cont}_{K}(M_{K},A)$,
\[
\xymatrix{
\mathcal{M}(\mathcal{L};A)\ar[r]^<<<<<{\sim}&\text{Hom}_{K}(M_{K},A)\\
\mathcal{M}^{\sigma}(\mathcal{L};A)\ar[r]\ar@{^{(}->}[u]\ar@{}[ur]|\circlearrowleft&
\text{Cont}_{K}(M_{K},A)\ar@{^{(}->}[u]
}
\]
\item Finally, using Theorem \ref{orth-groe}
and a density
argument, it is easy to prove the result for
$\sigma$-additive measures. Indeed,
let $B$ be a $\sigma$-orthogonal generating
set for the action of $G$.
The set $\mathcal{M}^\sigma(\mathcal{L};A,G)$
is in bijection with functions $\nu:B/N_G(B)\rightarrow A$, such that
\[
\nu(\bigvee_{i=1}^\infty b_i)=\sum_{i=1}^\infty \nu(b_i),\quad
\bigvee_{i=1}^\infty b_i, b_j\in B,\,b_i\bot b_j\,\forall i\neq j.
\]
\end{enumerate}
\end{remark}

Let us consider an example for what we have just proved. The above
theorem (and the accompanying remark) has the following consequence: in order to
define a measure on the orthomodular lattice of orthogonal
projectors $\mathcal{L}(\mathcal{W})$ of a von Neumann algebra
$\mathcal{W}$, it is sufficient to know an orthogonal generating set
for the projectors in $\mathcal{W}$. For the case of Type I factors,
this generating set is given by  one dimensional projections. We
will come back to this point later.

\

\begin{remark}
It is important to mention that, in general, it is a difficult
task to prove that a set $B\subseteq \mathcal{L}$ is
an orthogonal generating set for the action of a group $G$.
But, in some cases, it is possible to check that the
set $G\cdot B\subseteq \mathcal{L}$ is an orthogonal generating set.
Using this hypothesis for $B$ let us prove a weak
version of the Non-Boolean Groemer's integral theorem:
\begin{theorem}[Weak Non-Boolean Groemer's integral theorem]
The set $\mathcal{M}(\mathcal{L};A,G)$
is included in the set of $N_G(B)$-invariant
functions $\nu:B\to A$ such that
\[
\nu(b_1\vee b_2)=\nu(b_1)+\nu(b_2),\quad\forall
b_1,b_2\in B,\,b_1\bot b_2.
\]
\end{theorem}
\begin{proof}
By hypothesis, the multiplication map
$\mu:\mathbb{Z}^{\oplus B}\otimes_{\mathbb{Z}[N_G(B)]}\mathbb{Z}[G]\to M$
is surjective. Applying $(-)_G$,
the map $\tilde{\mu}:(\mathbb{Z}^{\oplus B})_{N_G(B)}\to M_G$ is also surjective.
Notice that $\tilde{\mu}(\overline{R})=0$, where $R\subseteq\mathbb{Z}^{\oplus B}$,
is generated by $b_1\vee b_2-b_1-b_2$ for all $b_1,b_2\in B$, $b_1\bot b_2$
and $\overline{R}=R_{N_G(B)}$. Then,
\[
(\mathbb{Z}^{\oplus B}/R)_{N_G(B)}\twoheadrightarrow M_G\Longrightarrow
\text{Hom}_{\mathbb{Z}}(M_G,A)\subseteq
\text{Hom}_{\mathbb{Z}}((\mathbb{Z}^{\oplus B})_{N_G(B)}/\overline{R},A).
\]
\end{proof}

\end{remark}

\subsection{Positive measures on orthocomplemented lattices.}

Now, we restrict ourselves to the study of positive measures on
orthocomplemented lattices.

\begin{defs}
Let $\mathcal{L}$ be an orthocomplemented lattice. A \emph{positive
measure} on $\mathcal{L}$ is a measure taking positive values in
$\mathbb{R}$. We use the following notation

\begin{itemize}
\item $\mathcal{M}_+(\mathcal{L})$ denotes the set of positive measures,
\item $\mathcal{M}_+(\mathcal{L};G)$ denotes the set of invariant positive measures,
\item $\mathcal{M}_+^\sigma(\mathcal{L})$ denotes the set of $\sigma$-positive measures,
\item $\mathcal{M}_+^\sigma(\mathcal{L};G)$ denotes the set of invariant $\sigma$-positive measures.
\end{itemize}

\end{defs}

\noindent We discuss next some examples of generating sets and their
geometric properties. But first, let us give a geometric
construction

\begin{defs}
Given a subset $C$ of a real vector space $V$
we define its \emph{dual cone} as,
\[
C^{\ast}=\{\ell:V\rightarrow \mathbb{R}\,\colon\,\ell(C)\geq 0\}=
\{H\subseteq V\,\colon\,C\subseteq H^+\}.
\]

\noindent From \cite[\S 2.6.1]{MR2061575}, we know that the double
dual cone of $C$ is its convex hull and it determines the same dual cone.
\end{defs}

\begin{example}
\textbf{Geometrical properties of generating sets.} Let $B$ be an
orthogonal generating set for $\mathcal{L}$. Let us denote by
$M_{\mathbb{R}}$  the real vector space
$M(\mathcal{L})\otimes_{\mathbb{Z}}\mathbb{R}$. The kernel of a non-zero linear
functional $M_{\mathbb{R}}\rightarrow\mathbb{R}$ is an hyperplane
$H$ inside $M_{\mathbb{R}}$. By appealing to the universal real
measure we can identify every element $b\in B$ with an element of
$M_{\mathbb{R}}$. Thus, we can write $B\subseteq M_{\mathbb{R}}$
instead of $\pi_{\mathbb{R}}(B)\subseteq M_{\mathbb{R}}$.
Let us apply the \emph{dual cone} construction to $B\subseteq
M_{\mathbb{R}}$.

Using Theorem \ref{orth-groe}, it is easy to check that the dual cone of $B$ is the set of positive
measures
\[
\mathcal{M}_+(\mathcal{L})\cong B^\ast:=
\{ H\subseteq M_\mathbb{R}\,\colon\, B\subseteq H^+\}.
\]
Similarly, it follows
\begin{align*}
\mathcal{M}_+(\mathcal{L};G)&\cong
\{ H\subseteq M_\mathbb{R}\,\colon\, B\subseteq H^+,\,H=H^G\}\\
\mathcal{M}_+^\sigma(\mathcal{L})&\cong
\{ H\subseteq M_\mathbb{R}\,\colon\, B\subseteq H^+,\,H=\overline{H}\}\\
\mathcal{M}_+^\sigma(\mathcal{L};G)&\cong \{ H\subseteq
M_\mathbb{R}\,\colon\, B\subseteq H^+,\,H=H^G=\overline{H}\}.
\end{align*}
\end{example}

\begin{example}
\textbf{Hilbert space and standard quantum mechanics.} \noindent Let
$\mathcal{H}$ be a separable Hilbert space and consider the
orthomodular lattice of orthogonal projectors $\mathcal{L}$. Given
that any closed subspace $S$ possesses an orthonormal basis
$\{v_i\}_{i=1}^\infty$, the family of projectors of range $1$
(basically, one dimensional subspaces) is a $\sigma$-orthogonal
generating set. Notice that any projector of rank $1$ is given by
$|v\rangle\langle v|\in \mathcal{B}(\mathcal{H})$ for some $v\in
\mathcal{H}$. Thus, for an arbitrary separable Hilbert space, a
$\sigma$-orthogonal generating set is given by $B=\{|v\rangle\langle
v|\,\colon\,v\in\mathcal{H},\,\|v\|=1\}$. Notice that $\mathcal{L}$
is defined inside $\mathcal{B}(\mathcal{H})$ (see Proposition
\ref{l:embedded} below). More generally, if $\mathcal{W}$ is a von
Neumann algebra, we still have a similar picture. Recall that the
lattice associated to $\mathcal{W}$ is the lattice
$\mathcal{L}(\mathcal{W})$ of projectors in $\mathcal{W}$. Then, a
$\sigma$-orthogonal generating set $B$ will be some subset of
projectors on $\mathcal{W}$. In the previous case,
$\mathcal{W}=\mathcal{B}(\mathcal{H})$, the rank one projectors form
a $\sigma$-orthogonal generating set. In more general situations,
$B$ will be a set of projectors (not necessarily of rank one, ex.
Type III factor) which always exists.
\end{example}

\noindent Let us now turn to the following natural question
\cite[Ch. 5]{Hamhalter-QuantumMeasureTheory}:
\begin{center}
\emph{Does any real measure on $\mathcal{L}(\mathcal{W})$ extend to
a bounded linear functional on $\mathcal{W}$?}
\end{center}

\noindent The above question concerns the relationship between
$\mathcal{M}(\mathcal{L}(\mathcal{W});\mathbb{R})$ and
$\text{Cont}_\mathbb{C}(\mathcal{W};\mathbb{C})$. In order to
describe this problem in our terms, let us prove the following
fundamental fact:

\begin{prop}\label{l:embedded}
Let $\mathcal{W}$ be a von Neumann algebra on a separable Hilbert
space and let
$\mathcal{L}(\mathcal{W})$ be its orthomodular lattice of
projectors. Then, the universal complex measure is injective
\[
\xymatrix{
\mathcal{L}(\mathcal{W})\ar@{^(->}[r]^{i}
\ar@{^(->}[d]_{\pi_{\mathbb{C}}}&\mathcal{W}\\
M_{\mathbb{C}}\ar[ur]_{\overline{i}}&
}
\]
The map $\pi_\mathbb{C}$ is the universal complex
measure and the map $i$ is the inclusion
$\mathcal{L}(\mathcal{W})\subseteq\mathcal{W}$.
Moreover, the image of $\overline{i}$ is dense in $\mathcal{W}$.

Notice that $\overline{i}$ is not necessarily injective.
\end{prop}
\begin{proof}
The set $\mathcal{L}(\mathcal{W})$ of orthogonal projectors in
$\mathcal{W}$ forms an orthomodular lattice which can be viewed as a
subset of $\mathcal{W}$ (i.e.,
$\mathcal{L}(\mathcal{W})\subseteq\mathcal{W}$). Due to the fact
that any von Neumann algebra is a complex vector space with respect
to the sum and multiplication by scalars, the inclusion map
$i:\mathcal{L}(\mathcal{W})\hookrightarrow\mathcal{W}$ satisfies the
axioms of a measure with values in a complex vector space and then,
from our previous results, this inclusion factorizes through
$M_{\mathbb{C}}$. This means that there exists a unique map
$\overline{i}$ such that $i=\overline{i}\circ\pi_{\mathbb{C}}$. From
the fact that $i$ is an injection, it follows that $\pi_\mathbb{C}$
is also injective.

Now, let us recall a standard result of von Neumann algebras: the
complex vector space spanned by $\mathcal{L}(\mathcal{W})$ is
norm-dense in $\mathcal{W}$ (see for example Corollary 0.4.9(b) and
the comment that follows in \cite{MR866671}). Then, given that the
space spanned by $\mathcal{L}(\mathcal{W})$ is contained in the
image of $\overline{i}$, we can conclude that
$\overline{i}(M_{\mathbb{C}})$ is also dense in $\mathcal{W}$.
\end{proof}

\begin{cor}
Let $\mathcal{W}$ be a von Neumann algebra over a separable Hilbert
space such that $\overline{i}$ is injective
($\mathcal{L}(\mathcal{W})\subseteq M_{\mathbb{C}}
\subseteq\mathcal{W}$). Then, there exists a bijection between
$\sigma$-probabilities (resp. $\sigma$-positive measures) on
$\mathcal{L}(\mathcal{W})$ and normal states (resp. positive linear
functionals) on $\mathcal{W}$.
\end{cor}

\begin{proof}
As we mentioned in Section \ref{s:StatesAsMeasures}, states are
defined as positive complex-valued continuous linear functionals
$\mu:\mathcal{W}\rightarrow\mathbb{C}$. Moreover, the
restriction of $\mu$ to $\mathcal{L}(\mathcal{W})$ defines a
$\sigma$-probability.

Suppose that we have a $\sigma$-probability
$\nu:\mathcal{L}(\mathcal{W})\to[0,1]$.
Given that $[0,1]$ is included in $\mathbb{C}$,
we can assume that the codomain of $\nu$ is $\mathbb{C}$ and it follows now
that $\nu$ satisfies the axiom of a complex $\sigma$-measure.
Then, from what we have proved above, there exists a unique extension
to a continuous linear functional
$\overline{\nu}:M_\mathbb{C}\to\mathbb{C}$.
From the fact that $M_\mathbb{C}$ can be injected
as a dense subset inside $\mathcal{W}$,
we can extend $\overline{\nu}$ by continuity (in a unique way)
to $\tilde{\nu}:\mathcal{W}\to\mathbb{C}$. Diagrammatically,
\[
\xymatrix{
\mathcal{L}(\mathcal{W})\ar[d]_{\forall\nu}\ar@{^(->}[r]\ar[dr]
&M_\mathbb{C}\ar@{^(->}[r]\ar[d]_<<<<{\exists!\overline{\nu}}&\mathcal{W}\ar[dl]^{\exists!\tilde{\nu}}\\
[0,1]\ar@{^(->}[r]&\mathbb{C}
}
\]

It remains to check that $\tilde{\nu}$ is a normal state. Given that
a positive operator $A$ has a positive spectrum $sp(A)$
(\cite[p.3]{MR866671}), we can approximate the identity over $sp(A)$
by simple functions with positive coefficients. Hence, any positive
operator $A\in\mathcal{W}$ can be approximated by projectors in
$\mathcal{L}(\mathcal{W})$. In particular, given a continuous linear
map $\tilde{\nu}:\mathcal{W}\to\mathbb{C}$ such that
$\tilde{\nu}(P)\geq 0$ for all $P\in\mathcal{L}(\mathcal{W})$, it
follows that $\tilde{\nu}(A)\geq 0$ for all $A\geq 0$ (i.e.
$\tilde{\nu}$ is positive). The normality condition follows from the
$\sigma$-additivity condition on projection operators.

The bijection between $\sigma$-positive measures on
$\mathcal{L}(\mathcal{W})$ and positive linear functional on
$\mathcal{W}$ is similar replacing $[0,1]$ with $\mathbb{R}$.
\end{proof}

\

Let us now use these constructions to rephrase Gleason's theorem in our
terms. Gleason's theorem can be usually stated as follows (see
\cite[Th. 3.13]{MR1662485} or \cite{MR0096113}):

\

\begin{theorem*}[Gleason]
Let $\mathcal{H}$ be a Hilbert space with $\dim(\mathcal{H})\geq 3$.
Then the set of $\sigma$-positive measures $\nu$ on $\mathcal{L}(\mathcal{H})$
is in bijection with the set of
positive operators $T$ of the trace class such that
\[
\nu=\text{tr}(T-).
\]
\end{theorem*}

\

We can give a similar result appealing to Proposition \ref{l:embedded}
and its Corollary:

\begin{theorem}
Let $\mathcal{H}$ be a separable
Hilbert space such that $\mathcal{L}(\mathcal{H})\subseteq M_{\mathbb{C}}\subseteq\mathcal{B}(\mathcal{H})$.
Then the set of $\sigma$-positive measures $\nu$ on
$\mathcal{L}(\mathcal{H})$
is in bijection with the set of
positive operators $T$ of the trace class such that
\[
\nu=\text{tr}(T-).
\]
\end{theorem}
\begin{proof}
According to the previous Corollary taking $\mathcal{W}=\mathcal{B}(\mathcal{H})$,
a $\sigma$-positive measure is the same a positive linear
functional $\nu:\mathcal{B}(\mathcal{H})\to\mathbb{C}$.
As it is well known, measures of this kind are given by
$\text{tr}(T-)$ for some positive operator $T\in
\mathcal{B}(\mathcal{H})$. Then, $\sigma$-positive measures on
$\mathcal{L}(\mathcal{H})$ are in bijection with
functionals $\text{tr}(T-)$,
where $T$ is some positive operator.
Given that $I\in\mathcal{L}(\mathcal{H})$,
by imposing $\nu(I)=1$ (which is the normalization condition for any
quantum state), we obtain that $T$ is of the trace class.
\end{proof}

\

\section{Conclusions}\label{s:Conclusions}

In this paper we presented a characterization of states in
orthocomplemented lattices. This was done by appealing to a
non-commutative version of geometric probability theory based on
algebraic geometry techniques. Our theoretical framework incorporates invariant
states (under groups acting by automorphisms) in a natural way. Our
main result, given by Theorem \ref{orth-groe}, allows us to
characterize states in orthomodular lattices and invariant states as
well. Finally, we were able to recast  Gleason's theorem in these terms.

\

\end{document}